\documentclass{tSYS2e}

\usepackage{epstopdf}
\usepackage{subfigure}
\usepackage{algorithm2e}

\usepackage{color}

\usepackage[longnamesfirst,sort]{natbib}
\bibpunct[, ]{(}{)}{;}{a}{,}{,}

\newtheorem{lem}{Lemma}
\newtheorem{theorem}{Theorem}
\newtheorem{defn}{Definition}
\newtheorem{rem}{Remark}

\newtheorem{prob}{Problem Formulation}

\def\mc{\mathcal}

\begin{document}



\title{Sensor Selection Cost Optimization for Tracking Structurally Cyclic Systems: a P-Order Solution}

\author{
\name{M. Doostmohammadian \textsuperscript{a}\textsuperscript{b}$^{\ast}$\thanks{$^\ast$Corresponding author. Email: mrdoost@gmail.com, m.doostmohammadian@ictic.sharif.edu}
, H. Zarrabi\textsuperscript{c}, and H. R. Rabiee\textsuperscript{a}}
\affil{\textsuperscript{a}ICT Innovation Center for Advanced Information and Communication Technology, School of Computer Engineering, Sharif University of Technology;
\textsuperscript{b}Mechanical Engineering Department, Semnan University;
\textsuperscript{c}Iran Telecommunication Research Center (ITRC)}
}

\maketitle

\begin{abstract}
Measurements and sensing implementations impose certain cost in sensor networks. The sensor selection cost optimization is the problem of minimizing the sensing cost of monitoring a physical (or cyber-physical) system. Consider a given set of sensors tracking states of a dynamical system for estimation purposes. For each sensor assume different costs to measure different (realizable) states. The idea is to assign sensors to measure states such that the global cost is minimized. The number and selection of sensor measurements need to ensure the observability to track the dynamic state of the system with bounded estimation error. The main question we address is how to select the state measurements to minimize the cost while satisfying the observability conditions. Relaxing the observability condition for structurally cyclic systems, the main contribution is to propose a graph theoretic approach to solve the problem in \textit{polynomial time}. Note that, polynomial time algorithms are suitable for large-scale systems as their running time is upper-bounded by a polynomial expression in the size of input for the algorithm. We frame the problem as a linear sum assignment with solution complexity of $\mathcal{O}(m^3)$.
\end{abstract}


\section{Introduction} \label{secintro}
Sensors and sensing devices are widespread in everyday use and are involved in many aspect of human life. Nowadays, sensors are advanced beyond the physical world and even are introduced in online social networks.
The emerging notion of IoT and the so-called Trillion Sensors roadmap further motivates sensor and actuator implementation in many physical systems and cyber networks 
A few examples are: in ecosystems and environmental monitoring \citep{may1972ecology}, security and vulnerability of social networks \citep{pequito_gsip,jstsp14}, eHealth and epidemic monitoring \citep{nowzari2016epidemic}, Dynamic Line Rating (DLR) in smart power grids \citep{usman_smc:08,kar2015consensus+grid}, etc.
In these large-scale applications the cost of sensing is a challenge. The cost may represent energy consumption, the economic cost of sensors, and even the additive disturbance due to, for example, long distance communication in wireless sensor networks. The rapidly growing size of IoT and sensor networks motivates minimal cost sensor-placement solution for practical applications.

There exist different approaches toward sensor selection optimization. In \citep{boyd2009sensor}, authors study sensor selection for noise reduction. This work introduces combinatorial problem of selecting $k$ out of $m$ sensors to optimize the volume of probabilistic confidence ellipsoid containing measurement error by adopting  a convex relaxation.
Authors in \citep{pequito_gsip}, consider minimum sensor coverage for dynamic social inference. Their idea is to find minimum sensor collection to ensure generic social observability. The authors show the relaxation lies in \textit{set covering} category and is generally NP-hard\footnote{Note that, the NP-hard problems are believed to have no solution in time complexity upper-bounded by a polynomial function of the input parameters.} to solve.  In \citep{sinopoli2013network_obsrv} the source localization problem under observability constraints is addressed. The authors aim to find the minimal possible observers to exactly locate the source of infection/diffusion in a network. They state that this problem is NP-hard and propose approximations to solve the problem.
In another line of research, distributed optimization is discussed in \citep{wang2010control}, where dynamic feedback algorithms robust to disturbance is proposed to minimize certain cost function over a sensor network. Optimal sensor coverage with application to facility allocation is studied in \citep{MiadCons}. The authors propose a distributed deployment protocol as a local optimal solution in order to assign resources to group of mobile autonomous sensors under certain duty to capability constraints.
Minimal actuator placement ensuring controllability is discussed in \citep{jad2015minimal}. Authors provide P-order approximations to a generally NP-hard problem by considering control energy constraints. In  \citep{jstsp14}, the minimal \textit{number} of observers for distributed inference in social networks is discussed. Similarly, minimizing the number of actuators for structural controllability following specific rank constraints is addressed in \citep{commault2015single}. 

\begin{figure}[!t]
	\centering
	\includegraphics[width=2.4in]{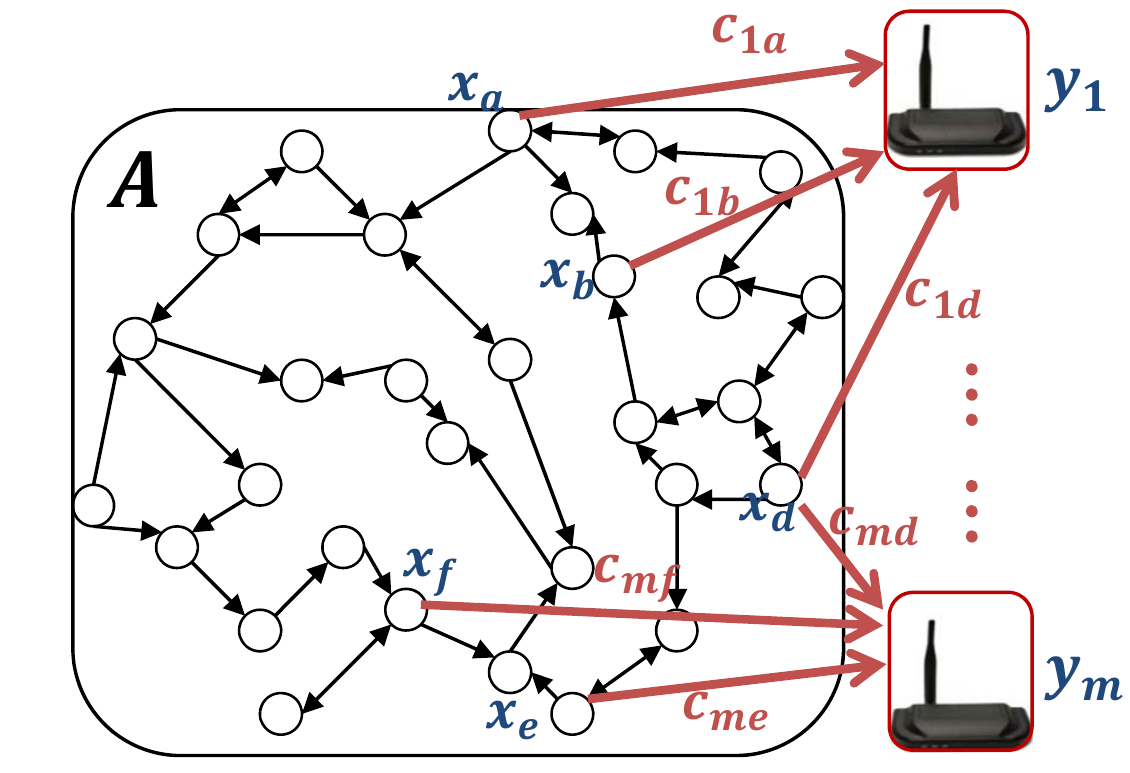}
	\caption{This figure shows a group of sensors monitoring a dynamical system,~$A$. Each sensor can measure states which are accessible/realizable. Assigning a sensor $y_i$ to measure a state $x_j$ has a cost $c_{ij}$. The sensor placement cost optimization problem finds the optimal sensor-state assignment such that the cost is minimum and the system is globally observable for the group of sensors.}
	\label{figsensor}
\end{figure}
This paper studies minimum sensor placement cost for tracking structurally cyclic dynamical systems (see Fig.\ref{figsensor}). In general, state measurements are costly and these costs may change for different sensor selection. This is due to various factors, e.g. sensor range and calibration, measurement accuracy, embedding/installation cost, and even environmental conditions. Therefore, any collection of state-sensor pairs may impose specific sensing costs. The main constraint, however, is that not all collection of sensor measurements provide an \textit{observable} estimation. Observability determines if the system outputs convey sufficient information over time to infer the internal states of the dynamic system. With no observability, no stable estimation can be achieved and the tracking error covariance grows unbounded.  In conventional sense, observability requires algebraic tests, e.g. the observability Gramian formulation \citep{bay} or the Popov-Belevitch-Hautus test \citep{hautus}. These methods are computationally inefficient especially in large-scale systems. In contrast, this paper adopts a structural approach towards observability. The methodology is irrespective of numerical values of system parameters, while the structure is fixed but the system parameters vary in time \citep{woude:03}. Indeed, this approach makes our solution practically feasible for Linear Structure-Invariant (LSI) systems. This may arise in linearization of nonlinear dynamics, where the linearized Jacobian is LSI while the values depend on the linearization point. As it is known the observability and controllability of Jacobian linearization is sufficient for observability and controllability of the original nonlinear dynamics \citep{Liu-nature,nonlin}.

Another structural property of the system is the system rank. In particular, for full rank systems the adjacency graph of the system (system digraph) is structurally cyclic, i.e. there exist a disjoint union of cycles covering all state nodes in system digraph. An example is dynamic system digraphs including self loops, which implies that for every system state the $1$st-derivative is a function of the same state (referred to as self-damped systems in \citep{acc13_mesbahi}). An example arises in systems representing ecological interactions \citep{may1972ecology,may2001book} among species, where intrinsic self-dampening ensures the eco-stability around equilibrium state.  In \citep{slotine2015intrinsic}, authors extended the case to network of coupled $d$th-order intrinsic dynamics. In such case, the self-dynamic impose a cyclic subgraph in the large-scale structure of the system digraph. In addition, in distributed estimation and sensor network literature typically it is assumed that the system matrix is invertible \citep{battistelli_cdc} and therefore system is full-rank. In \citep{liu-pnas,arcak2006diagonal} authors study the stability and observability of cyclic interconnected systems resulted from biochemical reactions, while in \citep{nowzari2016epidemic} authors analyze the self-damped epidemic equations integrated in social and human networks.
These examples motivates the study of structurally cyclic systems in this paper.

\textit{Contributions:}
Towards sensor selection cost optimization, this paper, first, considers LSI dynamic systems. In such systems the parameters vary in time while the system structure is unchanged as in Jacobian linearization of nonlinear dynamic systems. This is also the case, for example, in social systems with invariant social interactions, and in power systems with fixed system structure but time-varying parameters due to dynamic loading. Second, the observability constraint for cost optimization is framed as a selection problem from a necessary set of states. We relax the observability constraint by defining the \textit{equivalent} set of states necessary for observability. This is a novel approach towards cost optimization for estimation purposes. Third, the optimization is characterized as a Linear Sum Assignment Problem (LSAP), where the solution is of polynomial complexity. In this direction, the NP-hard observability optimization problem, as reviewed in the beginning of this section, is relaxed to a P-order problem for the case of \textit{structurally cyclic} systems. Note that, for general systems, this problem is NP-hard to solve (see \citep{pequito_gsip} for example). The relaxation in this paper is by introducing the concept of cost for states measured by sensors and the concept of structural observability. Further, the definition of the cost is introduced as a general mathematical concept with possible interpretation for variety of applications. Particularly, note that this P-order formulation is not ideal but practical, and this work finds application in monitoring large-scale systems such as social systems \citep{pequito_gsip,jstsp14}, eco-systems \citep{may1972ecology}, and even epidemic monitoring \citep{nowzari2016epidemic}. To the best of our knowledge, no general P-order solution is proposed in the literature for this problem.

\textit{Assumptions:}
The following assumptions hold in this paper:

(i) The system is globally observable to the group of sensors.

(ii) Number of sensors is at least equal to the number of crucial states necessary for observability, and at least one state is accessible/measurable by each sensor.

Without the first assumption no estimation scheme works, and there is no solution for optimal observability problem. In the second assumption, any sensor with no access to a necessary state observation is not a player in the optimization game. Other assumptions are discussed in the body of the paper.

The outline of the paper is as follows. In Section~\ref{secstruc}, relevant structural system properties and algorithms are reviewed. Section~\ref{secgraph} states the graph theoretic approach towards observability.  Section~\ref{seccost} provides the novel formulation for the cost optimization problem. Section~\ref{secLSAP} reviews the so-called assignment problem as the solution. Section~\ref{secrem} states some remarks on the results, motivation, and application of this optimal sensor selection scheme. Section~\ref{secexamp} illustrates the results by  two academic examples.
Finally, Section~\ref{secconc} concludes the paper.

\textit{Notation:} We provided a table of notations in Table~\ref{tab_realnet} to explain the terminologies and symbols in the paper.
\begin{table}[hbpt!]
	\centering
	\caption{ Table of Notation.}
	\begin{tabular}{|c|c|}
		\hline
		 $A$ &~system matrix\\
		\hline
		 $\mc{A}$ &~ structured system matrix \\
		\hline
		$H$ &~measurement matrix \\
		\hline
		$\mc{H}$ &~structured measurement matrix	\\
		\hline
		$\nu$ &~system noise\\
		\hline
		$\eta$ &~ measurement noise\\		
		\hline
		$k$ &~ discrete time index \\
		\hline
		$x$ &~ system state \\
		\hline
		$y$ &~ measurement	\\	
		\hline
		$c$ &~ state-sensor cost matrix \\
		\hline
		$\mc{C}$ &~ SCC-sensor cost matrix \\
		\hline
		$\mc{Z}$ &~ assignment matrix \\
		\hline
		$\tilde{c}$ &~ pseudo-cost	\\	
		\hline
		$n$ &~ number of states \\
		\hline
		$m$ &~ number of sensors/measurements	\\
		\hline
		$\mc{J}$ &~ Jacobian matrix \\
		\hline
		$\mc{X}$ &~ set of state nodes \\
		\hline
		$\mc{Y}$ &~ set of sensor nodes \\		
    	\hline
		\hline	
	\end{tabular}
	\label{tab_realnet}
\end{table}

\section{Structured System Theory} \label{secstruc}
Consider the state of linear system, $\underline{x}$, evolving as\footnote{The underline notation represents a \textit{vector} variable.}:
\begin{eqnarray} \label{sysc}
\dot{\underline{x}} = A\underline{x} + \underline{\nu}
\end{eqnarray}
and in discrete time as:
\begin{eqnarray} \label{sysd}
\underline{x}(k+1) = A\underline{x}(k) + \underline{\nu}(k)
\end{eqnarray}
where $\underline{x} \in \mathbb{R}^n $ is the vector of system states, and $\underline{\nu} \sim \mathcal{N}(0,V)$ is independent identically distributed (iid) system noise.
Consider a group of sensors, indexed by $y_i,~ i=1,\hdots,m$ each taking a noise-corrupted state measurement as:
\begin{eqnarray} \label{sensc}
y_i = H_i\underline{x} + \eta_i
\end{eqnarray}
or in discrete time as:
\begin{eqnarray} \label{sensd}
y_i(k) = H_i\underline{x}(k) + \eta_i(k)
\end{eqnarray}
where $H_i$ is a row vector, $y_i \in \mathbb{R}$ is the sensor measurement, and $\eta_i \sim \mathcal{N}(0,Q_i) $ is the zero-mean measurement noise at sensor $i$.

Let $\mathcal{A} \sim \{0,1\}^{n \times n}$ represent the structured matrix, i.e. the zero-nonzero pattern of the  matrix $A$. A nonzero element implies a system parameter that may change by time, and the zeros are the fixed zeros of the system. Similarly, $\mathcal{H} \sim \{0,1\}^{m \times n}$ represents the structure of measurement matrix $H$. A nonzero entry in each row of $\mathcal{H}$ represents the index of the measured state by the corresponding sensor. This zero-nonzero structure can be represented as a directed graph  $\mathcal{G}_{sys} \sim (\mathcal{X} \cup \mathcal{Y},\mathcal{E})$ (known as system digraph). Here, $\mathcal{X}$ is the set of state nodes $\{{x}_1,\hdots {x}_n\}$ each representing a state, and $\mathcal{Y}$ is the output set $\{{y}_1, \hdots {y}_m\}$ representing the set of sensor measurements. The nonzero entry $\mathcal{A}_{ij}$  is modeled by an edge ${x}_j \rightarrow {x}_i $. The set $\mathcal{E} = (\mathcal{X} \times \mathcal{X}) \cup  (\mathcal{X} \times  \mathcal{Y})$ is the edge set.  Edges in $\mathcal{E}_{xx} = \mathcal{X} \times \mathcal{X}$ represent the dynamic interactions of states in $\mathcal{G}_{A} \sim (\mathcal{X},\mathcal{E}_{xx})$, and edges $\mathcal{E}_{xy} = (\mathcal{X} \times  \mathcal{Y})$ in $\mathcal{G}_{xy} = (\mathcal{X} \cup \mathcal{Y},\mathcal{E}_{xy})$ represent the flow of state measurement information into sensors. It is clear that $\mathcal{G}_{sys} = \mathcal{G}_{A} \cup \mathcal{G}_{xy}$. Define a path as a chain of non-repeated edge-connected nodes and denote $\xrightarrow{path} \mathcal{Y}$ as a path ending in a sensor node in $\mathcal{Y}$. Define a cycle as a path starting and ending at the same node.

Similar to the structured matrices and the associated digraphs for linear systems, one can define a digraph for the Jacobian linearization of the nonlinear systems, also referred to as inference diagrams \citep{liu-pnas}. In the nonlinear case the structure of the system digraph is related to the the zero-nonzero structure of the Jacobian matrix $\mathcal{J}$. For system of equations $\dot{\underline{x}} = f(\underline{x}, \underline{\nu})$ if $\mathcal{J}_{ij} = \frac{\partial f_i}{\partial x_j}$ is not a fixed zero, draw a link ${x}_j \rightarrow {x}_i $ in the system digraph. This implies that ${x}_i$ is a function of ${x}_j$ and state ${x}_j$ can be inferred by measuring ${x}_i$ over time. Following this scenario for all pairs of states and connecting the inference links the system digraph $\mathcal{G}_A$ is constructed. It should be mentioned that we assume the nonlinear function $f$ is globally \textit{Lipschitz}, and therefore the system of equation has a unique solution and the Jacobian matrix is defined at all operating points.

The properties of the system digraph and its zero-nonzero structure are closely tied with the generic system properties. Such properties are almost independent of values of the physical system parameters.  It is known that if these specific properties of the system hold for a choice of numerical values of free parameters, they hold for almost all choices of system parameters, where these system parameters are enclosed in nonzero entries of the system matrix. Therefore, the zero-nonzero structure of the system and the associated system digraph ensures sufficient information on such generic properties. In general, efficient structural algorithms are known to check these properties while the numerical approach might be NP-hard to solve \citep{woude:03}. An example of generic properties are structural rank ($\mathcal{S}$-rank) and structural observability and controllability \citep{woude:03,jstsp14}.

\subsection{Structurally Cyclic Systems}
The following definition defines structurally cyclic systems:

\begin{defn}\label{defrank}
	A system is structurally cyclic if and only if its associated system digraph includes disjoint family of cycles spanning all nodes \citep{woude-rank}.
\end{defn}

There exist many real-world systems which are structurally cyclic. As mentioned in Section \ref{secintro}, any complex network/system governed by coupled $d$th-order differential equations and randomly weighted system parameters is structurally cyclic (see \citep{slotine2015intrinsic} for more information). Such structures may arise in biochemical reaction networks \citep{arcak2006diagonal,liu-pnas}, epidemic spread in networks \citep{nowzari2016epidemic}, ecosystems \citep{may1972ecology,may2001book} and even in social networks where each agent has intrinsic self-dampening dynamics represented as self-loop in social digraph (see the example in \citep{pequito_gsip}).

There are efficient methods to check if the graph is cyclic and includes a disjoint cycle family, namely matching algorithms. A matching of size $m$, denoted by $\mathcal{M}_m$, is a subset of nonadjacent edges in $\mathcal{E}_{xx}$ spanning $m$ nodes in $\mathcal{X}$. Define nonadjacent directed edges as two edges not sharing an end node. Define a maximum matching as the matching of maximum size in $\mathcal{G}_A$, where the size of the matching, $m$, is defined by the number of nodes covered. A perfect matching is a matching covering all nodes in the graph, i.e. $\mathcal{M}_n$ where $n = |\mathcal{X}|$.
\begin{lem}
	A system of $n$ state nodes is structurally cyclic if and only if its digraph includes a perfect matching $\mathcal{M}_n$. 
\end{lem}
\begin{proof}
	The detailed proof is given in \citep{murota}.
\end{proof}

The size of the maximum matching is known to be related to the structural rank of the system matrix defined as follows:
\begin{defn}
	For a structured matrix $\mathcal{A}$, define its structural rank ($\mathcal{S}$-rank) as the maximum rank of the matrix $A$ for all values of non-zero parameters. 
\end{defn}
Note that, the $\mathcal{S}$-rank of $\mathcal{A}$ equals to the maximum size of disjoint cycle family, where the size represents the number of nodes covered by the cycle family \citep{harary}. For structurally cyclic systems the maximum size of the cycle family equals to $n$, number of system states. This implies that for structurally cyclic system $\mathcal{S}\mbox{-rank}(\mathcal{A}) = n$. This result can be extended to nonlinear systems as $\mathcal{S}\mbox{-rank}(\mathcal{J}) = n$ at almost all system operating points. 

As an example, consider a graph having a random-weighted self-cycle at every state node. In this example, every self-cycle is a matching edge, graph contains a perfect matching and therefore is cyclic. On the other hand, these self-loops imply that every diagonal element in the associated structured matrix, $\mathcal{A}$, is nonzero. Having random values at diagonal entries and other nonzero parameters the determinant is (almost) always nonzero and system is structurally full rank \footnote{This can be checked simply by MATLAB, considering random entries as nonzero parameters of the matrix. The probability of having zero determinant is zero.}. This is generally true for network of intrinsic $d$th-order self dynamics (instead of $1$st order self-loops) as addressed in \citep{slotine2015intrinsic}. In the coupled dynamic equations a $d$th-order dynamic represents a cyclic component in system digraph. The intra-connection of these individual dynamics construct the digraph of large-scale physical system. Assuming time-varying system parameters the system remains structurally full rank.

\section{Graph Theoretic Observability} \label{secgraph}
Observability plays a key role in estimation and filtering. Given a set of system measurements, observability quantifies the information inferred from these measurements to estimate the global state of the system. This is irrespective of the type of filtering and holds for any estimation process by a group of sensors/estimators. Despite the algebraic nature of this concept, this paper adopts a graph theoretic approach towards observability. This approach is referred to as \textit{structural observability} and deals with system digraphs rather than the algebraic Gramian-based method. The main theorem on structural observability is recalled here.
\begin{theorem} \label{thmstruc}
	A system digraph is structurally observable if the following two conditions are satisfied: 
	
	\begin{itemize}
		
		\item Every state, is connected to a sensor via a directed path of states, i.e. $x_i \xrightarrow{path} \mathcal{Y}$, $i \in \{1,\hdots,n\}$.
		
		\item There is a sub-graph of disjoint cycles and output-connected paths that spans all state nodes.
		
	\end{itemize}
\end{theorem}
\begin{proof}
	The original proof of the theorem is available in the work by Lin \citep{lin} for the dual problem of structural controllability, and more detailed proof is available in \citep{rein_book}. The proof for structural observability is given in \citep{liu-pnas}.
\end{proof}
The conditions in Theorem~\ref{thmstruc} are closely related to certain properties in digraphs. The second condition holds for structurally cyclic systems, since all states are included in a disjoint family of cycles. The first condition can be checked by finding Strongly Connected Components (SCCs) in the system digraph \citep{asilomar11}. Recall that a SCC includes all states mutually reachable via a directed path. Therefore, the output-connectivity of any state in SCC implies the output-connectivity of all states in that SCC, and consequently, this satisfies the first condition in Theorem~\ref{thmstruc}. By measuring one state in every SCC, all states in that SCC are reachable (and observable), i.e. $x_i \in SCC_l$ and $x_i\xrightarrow{path} \mathcal{Y}$ implies $SCC_l \xrightarrow{path} \mathcal{Y}$. This further inspires the concept of \textit{equivalent} measurement sets for observability stated in the following lemma.
\begin{lem}
	States sharing an SCC are equivalent in terms of observability.
\end{lem}
\begin{proof}
	This is directly follows from the definition. Since for every two states $x_i$ and $x_j$ we have $x_j \xrightarrow{path} x_i$, then following Theorem~\ref{thmstruc} having $x_i \xrightarrow{path} \mathcal{Y}$ implies that $x_j \xrightarrow{path} \mathcal{Y}$. See more details in the previous work by the first author \citep{asilomar11}.
\end{proof}
Observationally equivalent states provide a set of options for monitoring and estimation. This is of significant importance in reliability analysis of sensor networks. These equivalent options are practical in recovering the loss of observability in case of sensor/observer failure \citep{asilomar14}.
In order to explore states necessary for observability, we partition all SCCs in terms of their reachability by states in other SCCs.
\begin{defn}
	~
	\begin{itemize}
		
		\item Parent SCC: is a SCC with no outgoing edge to states in other SCCs, i.e. for all $x_i \in SCC_l$ there is no $x_j \notin SCC_l$ such that $x_i \rightarrow x_j$.
		
		\item Child SCC: is a non-parent SCC (a SCC having outgoing edges to other SCCs), i.e. there exist $x_i \in SCC_l$ and $x_j \notin SCC_l$ such that $x_i \rightarrow x_j$.
		
	\end{itemize}
\end{defn}

\begin{lem} \label{lemsep}
	Parent SCCs do not share any state node.
\end{lem}
\begin{proof}
	The above lemma is generally true for all SCCs. The proof is clear; if two components share a state node they in fact make a larger component. 
\end{proof}
Following the first condition in Theorem~\ref{thmstruc}, the given definitions inspire the notion of necessary set of equivalent states for observability.
\begin{lem}  \label{lemSCCp}
	At least one measurement/sensing from every parent SCC is necessary for observability.
\end{lem}
\begin{proof}  
	This is because the child SCCs are connected to parent SCCs via a direct edge or a directed path. Therefore, $\mathcal{Y}$-connectivity of parent SCCs implies $\mathcal{Y}$-connectivity of child SCCs. In other words, $ x_i \in SCC_l$, $x_i \rightarrow x_j$ and $x_j \in SCC_k \xrightarrow{path} \mathcal{Y}$ implies $SCC_l \xrightarrow{path} \mathcal{Y}$. See detailed proof in the previous work by the first author \citep{asilomar11}.
\end{proof}
This further implies that \textit{ the number of necessary sensors for observability equals to the number of parent SCCs in structurally cyclic systems}. In such scenario, it is required to assign a sensor for each parent SCC in order to satisfy the observability condition. For more details on SCC classification and equivalent set for observability refer to \citep{asilomar11,jstsp14}.

\section{Cost Optimization Formulation} \label{seccost}
In sensor-based applications every state measurement imposes certain cost. The cost may be due to, for example, maintenance and embedding expenses for sensor placement, energy consumption by sensors, sensor range and calibration, and even environmental condition such as humidity and temperature. In this section, we provide a novel formulation of the \textit{minimal cost} sensor selection problem accounting for different sensing costs to measure different states. Contrary to \citep{pequito_gsip}, the final formulation in this section has a polynomial order solution as it is discussed in Section~\ref{secLSAP}.

\begin{prob}
	Assume a group of sensors and a cost $c_{ij}$ for every sensor $y_i , i \in \{1, \hdots, m\}$ measuring state $x_j , j \in \{1, \hdots, n\}$. Given the cost matrix $c$, the sensor selection cost optimization problem is to minimize sensing cost for tracking the global state of the dynamical system \eqref{sysc} (or discrete time system \eqref{sysd}). Monitoring the global state requires observability conditions, leading to the following formulation:
	\begin{equation}
	\begin{aligned}
	\displaystyle\min \limits_{\mathcal{H}} ~~ & \sum_{i=1}^{m} \sum_{j=1}^{n} (c_{ij}\mathcal{H}_{ij}) \\
	\text{s.t.} ~~ & (A,H)-observability,\\
	~~ &  \mathcal{H}_{ij} \in \{0,1\}\\
	\end{aligned}
	\label{minfirst}
	\end{equation}
	where $A$ and $H$ are, respectively, system and measurement matrix, and $\mathcal{H}$ represents the $0-1$ structure of $H$.
\end{prob}
In this problem formulation, $ \mathcal{H}$ is the $0-1$ pattern of $H$, i.e. a nonzero element $\mathcal{H}_{ij}$ represents the measurement of state $x_j$ by sensor $y_i$. 
First, following the discussions in Section~\ref{secstruc}, the observability condition is relaxed to structural observability. 
\begin{equation}
\begin{aligned}
\displaystyle\min \limits_{\mathcal{H}} ~~&  \sum_{i=1}^{m} \sum_{j=1}^{n} (c_{ij}\mathcal{H}_{ij}) \\
\text{s.t.} ~~ & (\mathcal{A},\mathcal{H})-observability,\\
~~ &  \mathcal{H}_{ij} \in \{0,1\}\\
\end{aligned}
\label{min1}
\end{equation}
Notice that in this formulation $(\mathcal{A},\mathcal{H})$-observability implies the \textit{structural} observability of the pair $(A,H)$. Primarily assume that the number of sensors equals to the number of necessary measurements for structural observability. In control and estimation literature \citep{jstsp14, commault2015single}, this is addressed to find the minimal number of sensors/actuators .
This consideration is in order to minimize the cost. Notice that extra sensors impose extra sensing cost, or they take no measurements and play no role in estimation. Therefore, following the assumptions in Section~\ref{secintro}, the number of sensors, at first, is considered to be equal to the number of necessary measurements for observability (i.e. number of parent SCCs). This gives the following reformulation of the original problem.
\begin{prob}
	Considering minimum number of sensors for observability, the sensor selection cost optimization problem is in the following form:
	\begin{equation}
	\begin{aligned}
	\displaystyle\min \limits_{\mathcal{H}} ~~  & \sum_{i=1}^{m} \sum_{j=1}^{n} (c_{ij}\mathcal{H}_{ij}) \\
	\text{s.t.} ~~ &(\mathcal{A},\mathcal{H})-observability,\\
	~~ &  \sum_{i=1}^{m} \mathcal{H}_{ij} \leq 1\\
	~~&  \sum_{j=1}^{n} \mathcal{H}_{ij} = 1\\
	~~ &  \mathcal{H}_{ij} \in \{0,1\}\\
	\end{aligned}
	\label{min2}
	\end{equation}
\end{prob}
The added conditions do not change the problem. The constraint $\sum_{i=1}^{m} \mathcal{H}_{ij} \leq 1$ implies that all states are measured by at most one sensor, and $\sum_{j=1}^{n} \mathcal{H}_{ij} = 1$ implies that all sensors are responsible to take a state measurement. Notice that, in case of having, say $N$, sensors more than the  $m$ necessary sensors for observability, this condition changes to $ \sum_{j=1}^{N} \mathcal{H}_{ij} \leq 1 $ to consider the fact that some sensors are not assigned, i.e. they take no (necessary) measurement.

Next, we relax the observability condition following the results of Section~\ref{secgraph} for structurally cyclic systems. Revisiting the fact that parent SCCs are separate components from Lemma~\ref{lemsep}, the problem can be stated as assigning a group of sensors to a group of parent SCCs.
For this formulation, a new cost matrix $\mathcal{C}_{m \times m}$ is developed. Denote by $\mathcal{C}_{ij}$, the cost of assigning a parent set, $SCC_j$, to sensor $y_i$. Define this cost as the \textit{minimum} sensing cost of states in parent $SCC_j$:
\begin{equation} \label{eqSCCcost}
\mathcal{C}_{ij}= \min \{c_{il}\},~ x_l \in  SCC_j,~ i,j \in \{1, \hdots, m\}
\end{equation}
This formulation transforms matrix $c_{m \times n}$ to matrix $\mathcal{C}_{m \times m}$. This transfers the sensor-state cost matrix to a lower dimension cost matrix of sensors and  parent SCCs. Further, introduce new variable $\mathcal{Z}\sim \{0,1\}^{m \times m}$ as a structured matrix capturing the assignment of sensors to parent SCCs. Entry $\mathcal{Z}_{ij}$ implies sensor indexed $i$ having a state measurement of SCC indexed $j$, and consequently $ SCC_j \xrightarrow{path} y_i $. Recalling that sensing all parent SCCs guarantee observability (see Lemma~\ref{lemSCCp}), the problem formulation can be modified accordingly in a new setup as follows.
\begin{prob} \label{prob_final}
	For structurally cyclic systems, having a set of $m$ sensors to be assigned to $m$ parent SCCs, the sensor selection cost optimization is given by:
	\begin{equation}
	\begin{aligned}
	\displaystyle\min \limits_{\mathcal{Z}} ~~  & \sum_{i=1}^{m} \sum_{j=1}^{m} (\mathcal{C}_{ij}\mathcal{Z}_{ij}) \\
	\text{s.t.} ~~ &  \sum_{j=1}^{m} \mathcal{Z}_{ij} = 1 \\
	~~&  \sum_{i=1}^{m} \mathcal{Z}_{ij} = 1 \\
	~~ &  \mathcal{Z}_{ij} \in \{0,1\} \\
	\end{aligned}
	\label{minlsap}
	\end{equation}
\end{prob}
In this formulation, the new constraint $ \sum_{j=1}^{m} \mathcal{Z}_{ij} = 1$ is set to satisfy sensing of all parent SCCs as necessary condition for observability. The formulation in \eqref{minlsap} is well-known in combinatorial programming and optimization. It is referred to as Linear Sum Assignment Problem (LSAP) \citep{assignmentSurvey}. It is noteworthy that the three statements in this section represent the same problem and the differences stem from mathematical relaxations and observability consideration. 

The above formulation is one-to-one assignment of sensors and parent SCCs. By changing the first constraint to $\sum_{j=1}^{m} \mathcal{Z}_{ij} \geq 1$ we allow more than one parent SCC to be assigned to each sensor. This is the generalization to the primary assumption of assigning only one state to each sensor. For the second constraint in \eqref{minlsap}, considering $\sum_{i=1}^{m} \mathcal{Z}_{ij} \geq 1$ implies that more than one sensor may be assigned to a parent SCC. This adds redundancy in sensor selection and consequently increases the cost, and thus should be avoided. On the other hand, $\sum_{i=1}^{m} \mathcal{Z}_{ij} \leq 1$ violates the necessary condition for observability as some of the SCCs may not be assigned and tracked by sensors. 

\section{Linear Sum Assignment Problem (LSAP)} \label{secLSAP}
The novel formulation of sensor selection problem proposed in the Problem Formulation~\ref{prob_final} is known to be a classical optimization problem referred to as the \textit{Assignment problem}.
Assignment problem is widely studied as many problems, e.g. in network flow theory literature, are reduced to it. The problem deals with matching two sets of elements in order to optimize an objective function. Linear Sum Assignment Problem (LSAP) is the classical problem of assigning $m$ tasks to $m$ agents (or matching $m$ grooms with $m$ brides, $m$ machines/companies to $m$ jobs, etc.) such that the matching cost is optimized \citep{assignmentSurvey}. The LSAP is mathematically similar to the weighted matching problem in bipartite graphs. This problem is also called one-to-one assignment as compared to one-to-many assignment problem in which one agent is potentially assigned to more than one task. There have been many solutions to this problem. From the original non-polynomial solution to later polynomial-time primal-dual solutions including the well-known Hungarian method. The Hungarian Algorithm, proposed by Kuhn \citep{kuhnHungarian} and later improved by Munkres, is of complexity order of $\mathcal{O}(m^4)$ with $m$ as the number of tasks/agents. The algorithm was later improved by \citep{edmondsHungarian} to the complexity order of $\mathcal{O}(m^3)$. The algorithm is given in Algorithm~\ref{alg_hung}.

\begin{algorithm}[!t] \label{alg_hung} 
	\textbf{Given:} Cost matrix $\mathcal{C}=[\mathcal{C}_{ij}]$ \;
	\For{$i=1,\hdots,n$}{
		$u_i =$ smallest integer in row $i$ of $\mathcal{C}$\;   
		\For{$j=1,\hdots,n$}{$\hat{\mathcal{C}}_{ij}=\mathcal{C}_{ij}-u_i$\	
		}
	}
	\For{$j=1,\hdots,n$}{
		$v_j =$ smallest integer in column $j$ of $\hat{\mathcal{C}}$\;   
		\For{$i=1,\hdots,n$}{$\hat{\mathcal{C}}_{ij}=\hat{\mathcal{C}}_{ij}-v_j$\	
		}
	}
	$S =$ an independent set of zeros of max size in  $\hat{C}$\; 	 
	$q = |S|$ \;
	\While{$q<n$}{
		Cover $\hat{\mathcal{C}}$\;
		$k = $ smallest entry in $\hat{\mathcal{C}}$ not covered by a line\;
		\For{$i=1,\hdots,n$}{
			\For{$j=1,\hdots,n$}{\If{$\hat{\mathcal{C}}_{ij}$ is not covered}{$\hat{\mathcal{C}}_{ij} = \hat{\mathcal{C}}_{ij} - k$ }
				\If{$\hat{\mathcal{C}}_{ij}$ is covered twice}{$\hat{\mathcal{C}}_{ij} = \hat{\mathcal{C}}_{ij} + k$ }	
			}
		}	
		$S =$ an independent set of zeros of max size in  $\hat{\mathcal{C}}$\;
		$q = |S|$ \;}
	\For{$i=1,\hdots,n$}{
		\For{$j=1,\hdots,n$}{
			\If{$\hat{\mathcal{C}}_{ij} \in S$ }{$\mathcal{Z}_{ij} = 1$}
			\Else{$\mathcal{Z}_{ij} = 0$}
		}
	}
	
	\textbf{Return} $\mathcal{Z}=[\mathcal{Z}_{ij}]$\;\
	
	\caption{Hungarian Algorithm}
\end{algorithm}

Other than these original solutions, recently new linear programming methods to solve the classical one-to-one LSAP  and variations of this original setting are discussed. To name a few, \textit{distributed} assignment problem based on a game-theoretic approach is proposed in \citep{zavlanos2008distributed}. Sensors/Agents are assigned to tasks relying only on \textit{local} information of the cost matrix. The complexity of the algorithm is $\mathcal{O}(m^3)$ in the worst case scenario. In \citep{bertsekas1981assign} a new algorithm is proposed whose average complexity matches Edmonds Hungarian method in large-scale. 
All these solutions can be applied to solve the Problem Formulation~\ref{prob_final} and its variant, for example even when the sensing costs are changing. However, in terms of performance the Edmonds’ Hungarian algorithm \citep{edmondsHungarian} is more practical and used in programming softwares like MATLAB. The algorithm by \citep{zavlanos2008distributed} is practical in distributed setting while the algorithm by \citep{bertsekas1981assign} is as practical as \citep{edmondsHungarian} only in large-scale applications. Note that the focus of this paper is on the polynomial complexity of such algorithms to be practical in large-scale application. Therefore, although other non-polynomial solutions to LSAP may exist, they are not of interest  in large-scale sensor selection optimization.

Note that, in the LSAP the cost matrix has to be a complete $m$ by $m$ matrix. However, in practical application some states may not be measured by some sensors (not \textit{realizable} by some sensors). This may, for example, caused by mismatch in range/calibration of the sensor and what is required for the state measurement. In the sensor-state cost matrix, $c$, this simply implies that some entries are not defined. For this unmeasurable states, the cost is infinite, in application a large enough cost (pseudo-cost $\tilde{c}_{ij}$) can be given. By introducing $\tilde{c}_{ij}$ and having a complete cost matrix, the LSAP problem can be solved using anyone of the polynomial methods mentioned in this section.  Notice that if the optimal cost from LSAP in \eqref{minlsap} is greater than the pseudo cost $\tilde{c}_{ij}$, the sensor selection has no feasible solution. A possible explanation is that at least one parent SCC is not realizable by any sensor, implying the assignment of a pseudo cost by LSAP. In case the feasible solution exists, no non-realizable state is assigned and the LSAP gives the optimal feasible solution in polynomial time.

\section{Remarks} \label{secrem}
This section provides some remarks to further illustrate the results, motivation, and application of the polynomial order sensor selection solution proposed in this work.
\begin{rem}
The main motivation on this paper is to find a polynomial order solution to optimize sensor selection problem for cyclic systems. Notice that in general, as mentioned in the introduction and literature review, the problem is NP-hard to solve, see for example \citep{pequito_gsip} and references therein. However, we showed that if system is cyclic there exist a polynomial order solution for sensor selection optimization. Note that this is significant in large-scale system monitoring as polynomial order algorithms are practical in large-scale applications because their running time is upper-bounded by a polynomial expression in the size of input for the algorithm. Examples of such large-scale cyclic systems are given in the introduction. 
\end{rem}
\begin{rem} \label{remLSI}
It should be mentioned that LSI dynamics are practically used in state estimation and complex network literature, see \citep{Liu-nature,liu-pnas,jstsp14,woude:03,davison1973LSI,davison1974LSI,egerstedt2007LSI}  and references therein. The motivation behind structured system theory is that this approach holds for systems with time-varying parameters while the system structure is fixed. This is significant in system theory as in many applications the system non-zero parameters change in time while the zero-nonzero pattern of the system matrix is time-invariant. Indeed, this structural analysis deals with system properties (including observability and system rank) that do not depend on the numerical values of the parameters but only on the underlying structure (zeros and non-zeros) of the system \citep{woude:03,davison1973LSI,davison1974LSI}. It is known that  if a structural property holds for one admissible choice of non-zero elements/parameters it is true for almost all choices of non-zero elements/parameters and, therefore, is called generic property  \citep{woude:03}. Another motivation is in linearization of nonlinear systems where the nonlinear model is linearized over a continuum of operating points, see \citep{jstsp14,Liu-nature}. In this case the structure of the Jacobian matrix is fixed while the matrix elements change based on the linearization point, therefore implying the LSI system model. In this direction, the observability/controllability of LSI model implies the observability/controllability of the nonlinear model \citep{jstsp14,Liu-nature} and therefore the results of the LSI approach leads to conclusions on the nonlinear model.  
\end{rem} 
\begin{rem} Based on the mentioned features of LSI model in Remark~\ref{remLSI}, 
the structural observability almost always implies algebraic observability, therefore LSI relaxation in Problem Formulation 2 almost always holds. Further, for structurally cyclic systems the problem can be exactly framed as a LSAP, and therefore the relaxation in Problem Formulation 3 and polynomial order solution is exact for cyclic systems.
\end{rem}
\begin{rem} Note that the SCC decomposition is unique \citep{algorithm} and therefore the cost matrix $\mc{C}$ and the formulation in \eqref{minlsap} are uniquely defined.
\end{rem}
\begin{rem}
	While this work focuses on sensor selection and observability, the results can be easily extended to the dual problem of controllability and particularly input/actuator selection. In this case the problem is to choose among the possible inputs to direct/control the dynamical system to reach the desired state with optimal cost. Note that the only mathematical difference is that the constraint in Problem Formulation 1 and 2 changes to $(\mc{A},\mc{H})$-controllability and the same graph theoretic relaxation holds. Because of duality the problem changes to assigning sensors optimally to Child SCCs resulting the same formulation as in Problem Formulation 3 where the solution is known via Hungarian algorithm.    
	In fact, in the context of control of networked systems, this problem is also known as  the so-called \textit{leader selection}. In this problem using LSI model the idea is to determine the control leaders in structured multi-agent system. In this case the cost, for example,  may represent energy consumption by agents. For more information on this subject we refer interested readers to \citep{fitch2013leader,lin2011leader,lin2014leader,clark2014leader}.
\end{rem}

\section{Illustrative Examples} \label{secexamp}
This section provides academic examples to illustrate the results of the previous sections.

\textit{Example 1:}  Consider a dynamical system with the associated digraph given in Fig.\ref{figgraph}. Every node represents a state of the system and every edge represent the dynamic interaction of two states. For example an edge from $x_3$ to $x_1$ and the self-loop on $x_1$ implies $\dot{ x_1} = a_{13}x_3 + a_{11}x_1 $. Assuming a nonlinear dynamic system, the same link represents a possible nonlinear interaction function $\dot{ x_1} = f_1(x_3,x_1)$ where the Jacobian linearization is in the form $\dot{ x_1} = \frac{\partial f_1}{\partial x_3}x_3 + \frac{\partial f_1}{\partial x_1}x_1$.\footnote{Notice that having self-cycle at every node implies that the diagonal entries of the Jacobian matrix, $\frac{\partial f_i}{\partial x_i}$, are non-zero and the Jacobian is structurally full rank.} Such terminology holds for all state nodes and edges in the system digraph and relates the system digraph to the differential equation governing the dynamic phenomena.
\begin{figure}[!t]
	\centering
	\includegraphics[width=2.8in]{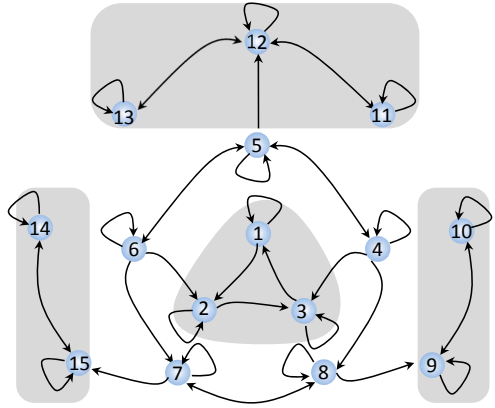}
	\caption{This figure shows a cyclic system digraph. Among the components in this graph, the highlighted ones represent parent SCCs. Every parent SCC must be tracked by (at least) one sensor.}
	\label{figgraph}
\end{figure}
\begin{figure}[!t]
	\centering
	\includegraphics[width=2.35in]{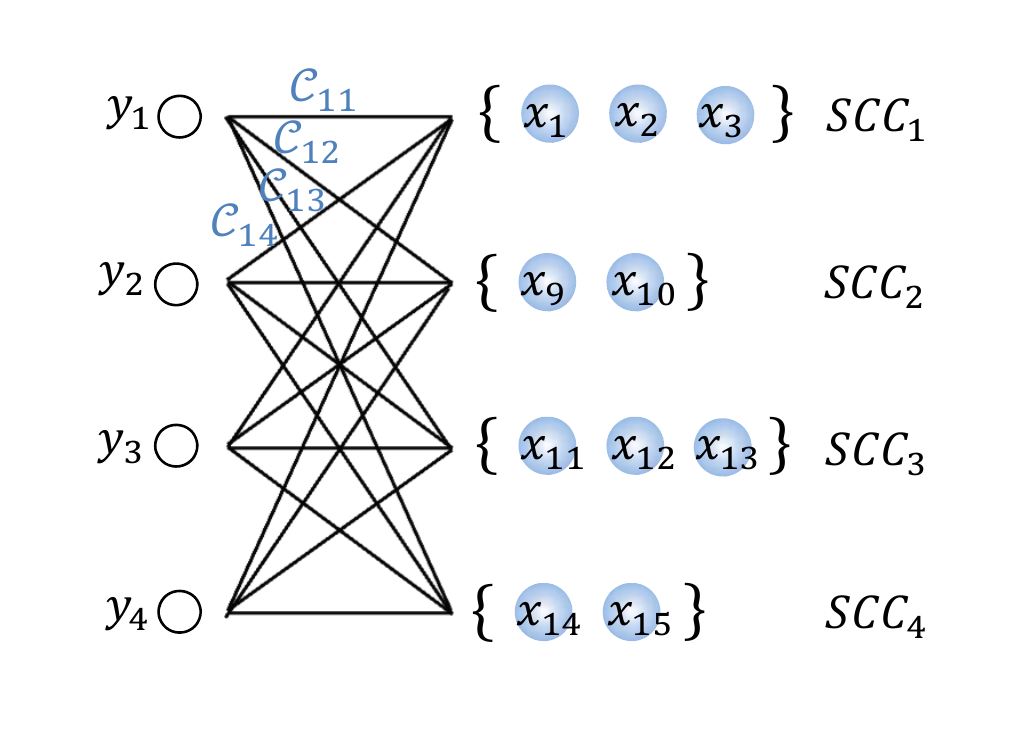}
	\caption{This figure shows the possible assignment of (states in) parent SCCs to sensors in Fig.\ref{figgraph}. The assignment cost for parent SCC-sensor pairs is defined as the minimum cost of states in each parent SCC according to \eqref{eqSCCcost}. Therefore the problem is framed as a linear sum assignment of SCCs and sensors as proposed in \eqref{minlsap}}.
	\label{figassign}
\end{figure}
A group of sensors are needed to track the $15$ states of this dynamic system. If a state is measurable by a sensor, the measurement is associated with a cost $c_{ij}$.  However, not every state is measurable by every sensor.  In this example, we assume the pseudo-cost as $\tilde{c}_{ij} = max(c_{ij})*m*n $. This prevents the assignment algorithm to assign these non-realizable sensor-state pairs. Among the realizable states, some states are not necessary to be measured. This is defined based on the SCC classification discussed in Section~\ref{secgraph}.
In the particular example of Fig.~\ref{figgraph}, the inner component and the outer components have no outgoing edges to other SCCs; therefore, $\{x_1,x_2,x_3\}$, $\{x_9,x_{10}\}$, $\{x_{11},x_{12},x_{13}\}$, and $\{x_{14},x_{15}\}$ are parent SCCs. The other components, $\{x_4,x_5,x_6\}$ and $\{x_7,x_8\}$, are child SCCs. For observability, each parent SCC is needed to be tracked by at least one sensor. The selection of which state to be measured in each parent SCC is cost-based; in every parent SCC, each sensor measures the state with minimum cost
. This takes the problem in the form given in Fig.\ref{figassign} and in the form of \eqref{minlsap}. Then, Hungarian method in Section~\ref{secLSAP} is applied to solve this LSAP.

For numerical simulation, in this system graph example we consider uniformly random costs $c_{ij}$ in the range $(0,10)$. Number of sensors equals to $m=4$ that is the number of parent SCCs. The non-realizable states are defined randomly with probability $50\%$, i.e. almost half of the states are not measurable by sensors \footnote{We should mention that this is only for the sake of simulation to check the algorithm. In real applications if the measurable states are not observable no sensor selection optimization algorithm can provide an observable estimation of the system, therefore in real applications it is usual to assume that at least one observable solution exist for the problem otherwise no sensor selection and estimation scheme works.}. In the assignment algorithm, the pseudo-cost of a non-realizable state is considered $\tilde{c}_{ij}=max\{c_{ij}\}*m*n$, which is certainly more than $\sum_{i=1}^{m} \sum_{j=1}^{n} c_{ij}$. Fig.~\ref{figsimulation} shows the cost of all realizable observable and non-observable assignments. In this figure, for the sake of clarification the indexes are sorted in ascending and descending cost order respectively for observable and non-observable assignments. Among the realizable state-sensor pairs, if the selected sensors do not measure one state in each parent SCC, this sensor selection is not observable (violating the Assumption (i) in  Section~\ref{secintro}). Among the observable selections, the optimal sensor selection has the minimum cost of $10.88$, which matches the output of the LSAP using the Hungarian method. Note that, the naive solution in Fig.~\ref{figsimulation} has complexity $\mc{O}(m!)$, and is only provided for clarification and checking the results of the LSAP solution.
\begin{figure}[!t]
	\centering
	\includegraphics[width=2.8in]{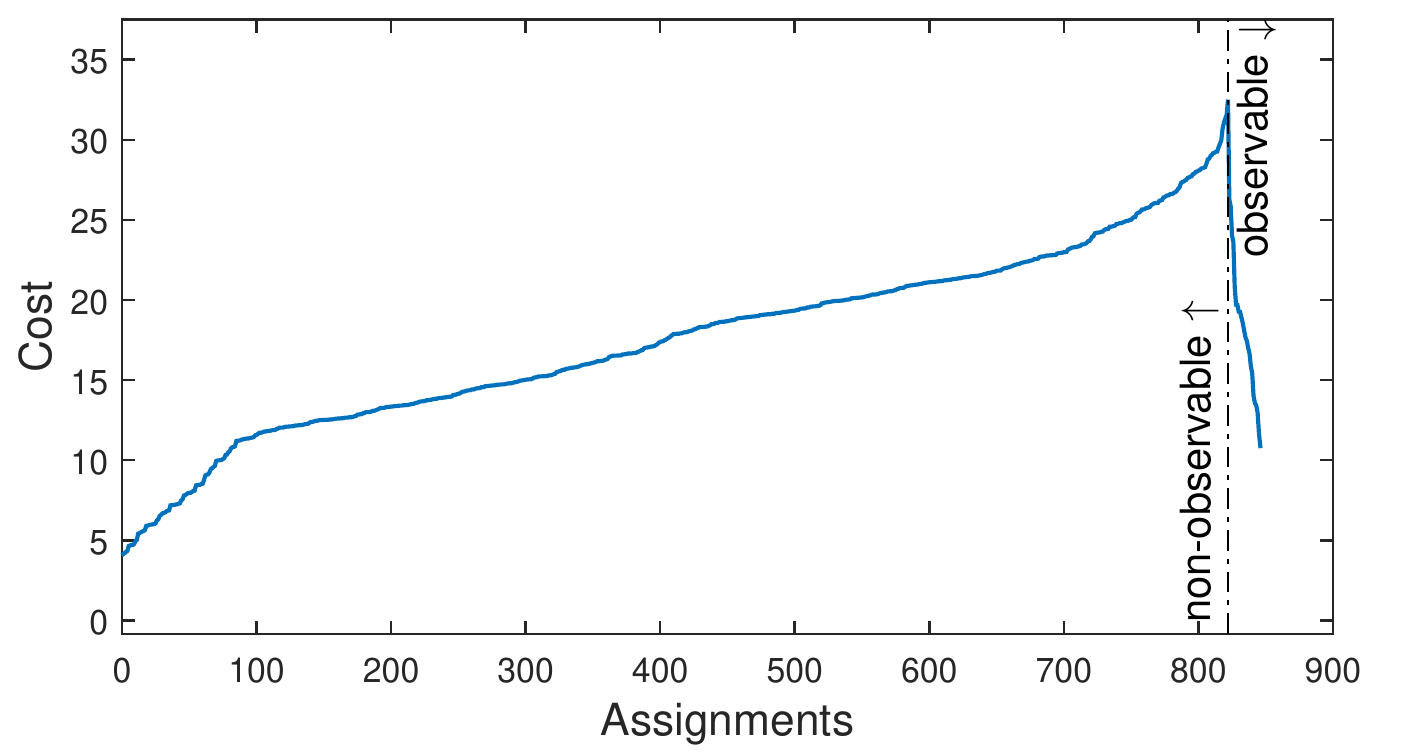}
	\caption{In order to check the results, this figure shows the sensing cost for all observable and non-observable sensor selections. The cost of optimal sensor assignment from the Hungarian method matches the minimum cost for observable selections (in the right-hand side of the graph). }
	\label{figsimulation}
\end{figure}

\textit{Example 2:}
In order to show the advantage of using the proposed approach  over the existing methods we provided another example. This example, shown in Fig.\ref{figgraph2}, is similar to the example given in \citep{pequito_gsip}, where the graph represents a social dynamic system. In such social digraphs each node represents an individual and the links represent social interaction and opinion dynamics among the individuals \citep{jstsp14,FriedkinSocial}. According to \citep{jstsp14,FriedkinSocial}, the social system is generally modeled as LSI system where the social interactions (as the structure) are fixed while the social influence of individuals on each-other change in time. This is a good example stating the motivation behind considering LSI model in this work. We intentionally presented this example to compare our results with \citep{pequito_gsip}. As \citep{pequito_gsip} claims for such example there is no polynomial order solution to solve the problem while here we present a sensor selection algorithm with polynomial complexity of $\mc{O}(m^3)$. Similar to the Example 1, a group of sensors (referred to as information gatherers in social system \citep{pequito_gsip}) are required to monitor $20$  states of the social system. Notice that having a cycle family covering all states the system is Structurally full-rank. The measurement of each state by each sensor is associated with a cost $c_{ij}$ and if not measurable the cost is assigned with $c'_{ij}=max\{c_{ij}\}*m*n$. Applying the DFS algorithm one can find the SCCs and Parent/Child classification in $\mc{O}(m^2)$ as shown in Fig.\ref{figgraph2}. Then, the assignment problem in Fig.\ref{figassign2} and in the form of \eqref{minlsap} can be solved using the Hungarian method in $\mc{O}(m^3)$.
\begin{figure}[!t]
	\centering
	\includegraphics[width=2.7in]{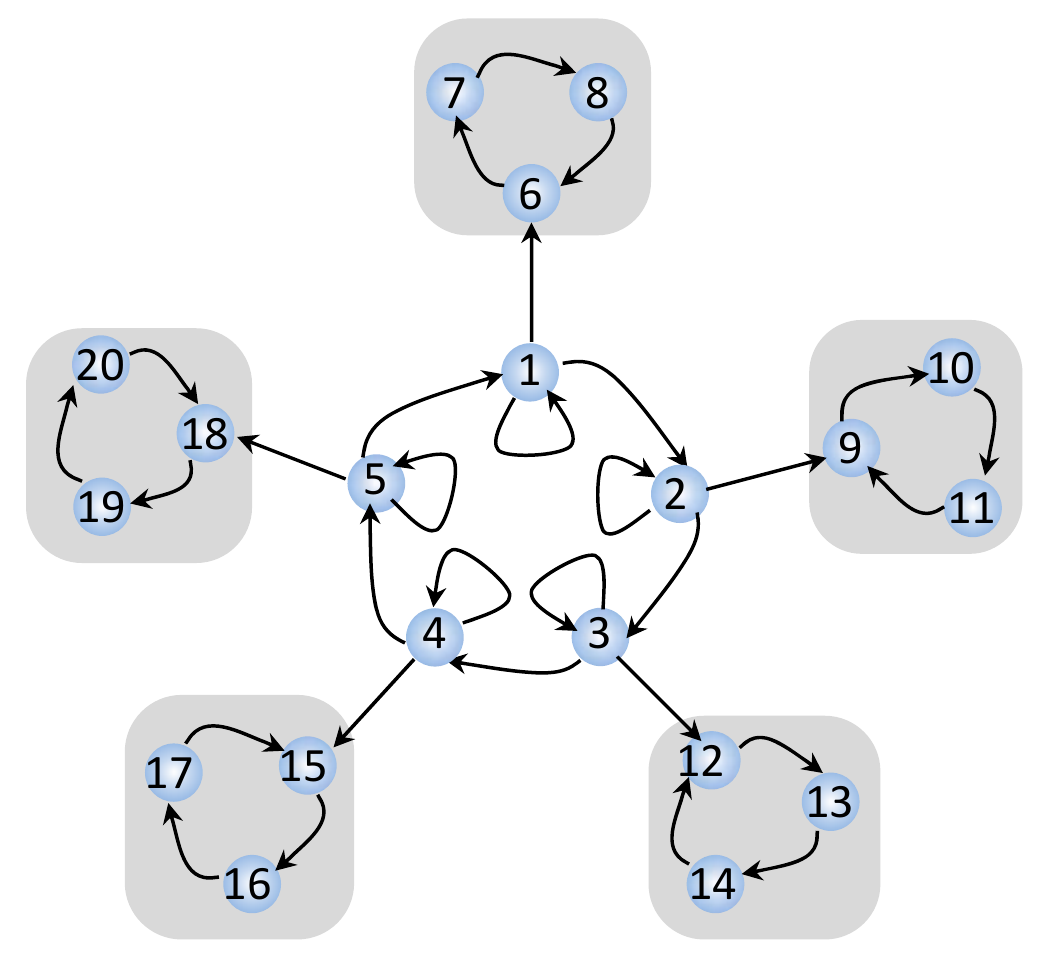}
	\caption{This figure shows a cyclic social system. The highlighted components represent parent SCCs, each tracked by a social sensor.}
	\label{figgraph2}
\end{figure}
\begin{figure}[!t]
	\centering
	\includegraphics[width=2.25in]{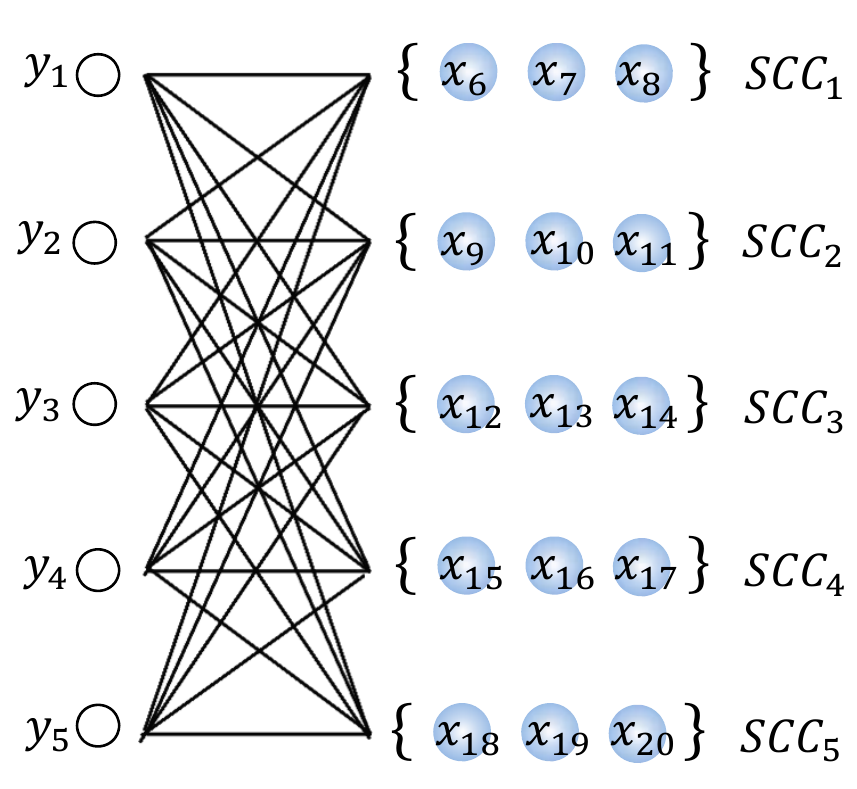}
	\caption{This figure shows the assignment of states in parent SCCs to social sensors for social system in Fig.\ref{figgraph2}, where the assignment cost is defined as the minimum cost of states in each parent SCC according to \eqref{eqSCCcost}. Therefore the problem is framed as a LSAP (see \eqref{minlsap}) and solvable by Hungarian method.}
	\label{figassign2}
\end{figure}
Again we consider uniformly random costs $c_{ij}$ in the range $(0,10)$ for numerical simulation. Since there are $5$ parent SCCs in this social digraph we need $m=5$ social sensor. Among the sensor-state pairs, the non-realizable states are defined randomly with probability $30\%$ with pseudo-cost $\tilde{c}_{ij}=max\{c_{ij}\}*m*n$. In Fig.~\ref{figsimulation2} the costs of all realizable observable and non-observable assignments are shown, where among the observable cases the optimal sensor selection has the minimum cost of $7.048$. As expected, this value matches the output of the proposed solution, i.e. the Hungarian algorithm for the LSAP  in equation \eqref{minlsap}. 
\begin{figure}[!t]
	\centering
	\includegraphics[width=2.8in]{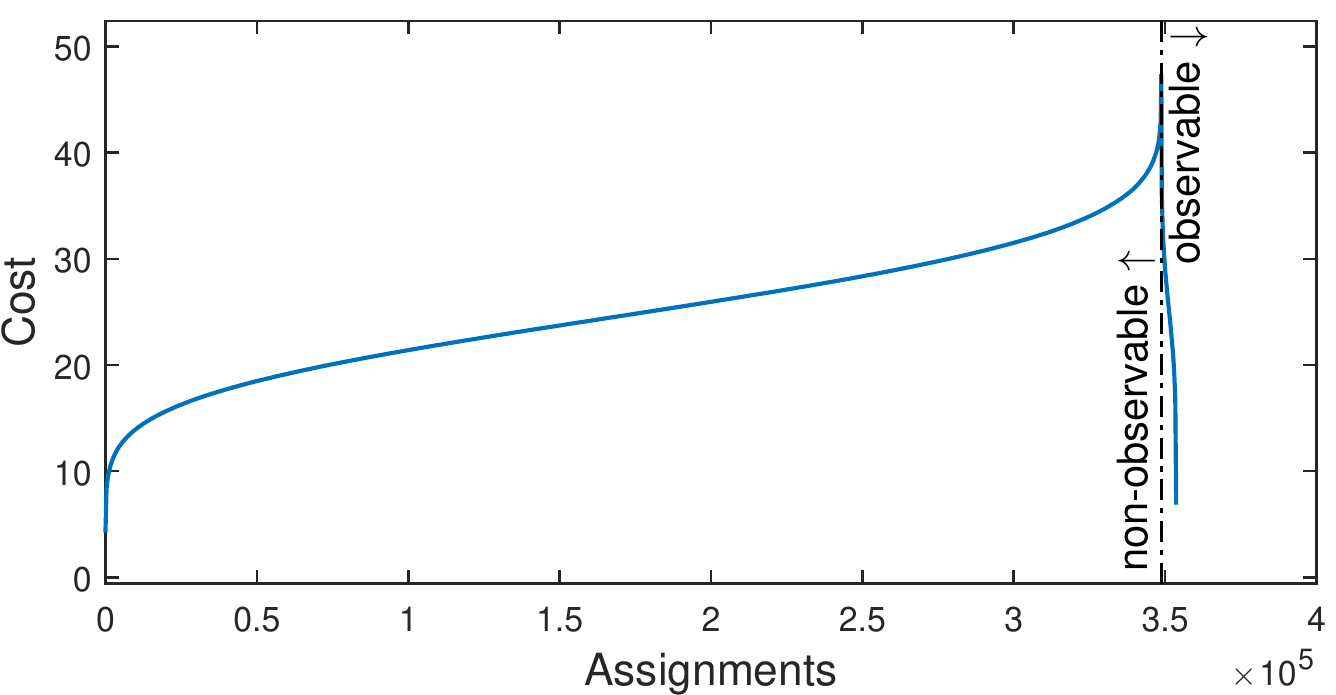}
	\caption{This figure shows the sensing cost for all observable and non-observable sensor selections for the system example given in Fig.\ref{figgraph2}. The optimal sensor selection cost of the Hungarian method solving the assignment in Fig.\ref{figassign2} matches the minimum cost for observable assignments. }
	\label{figsimulation2}
\end{figure}
 
\section{Conclusion} \label{secconc}
There exists many efficient algorithms to check the matching properties and structural rank of the system digraph, namely \textit{Hopcraft-Karp algorithm} \citep{hopcraft} or the \textit{Dulmage-Mendelsohn decomposition} \citep{dulmage58} of $\mc{O}(n^{2.5})$ complexity. Moreover, the SCC decomposition and the partial order of SCCs is efficiently done in running time $\mc{O}(n^2)$ by using the \textit{DFS algorithm} \citep{algorithm}, or the \textit{Kosaraju-Sharir algorithm} \citep{algorithm2}. As mentioned earlier in Section~\ref{secLSAP}, the LSAP solution is of complexity of $\mc{O}(m^3)$. This gives the total complexity of $\mc{O}(n^{2.5}+m^3)$ to solve the sensor coverage cost optimization problem. In dense graphs typically the nodes outnumber parent SCCs; assuming $m \ll n $, the complexity of the algorithm is reduced to $\mc{O}(n^{2.5})$. In case of knowing that the system is structurally cyclic, e.g. for self-damped systems, the running time of the solution is $\mc{O}(n^2)$.

In practical application using MATLAB, the \texttt{sprank} function checks the structural rank of the system (i.e. the size of maximum matching in the system digraph). System is structurally cyclic if \texttt{sprank(A)} equals to $n$, the size of the system matrix. To find the partial order of SCCs the straightforward way (but not as efficient) is to use \texttt{dmperm} function. This function takes the system matrix $A$ and returns the permutation vectors to transfer it to upper block triangular form and the boundary vectors for SCC classification. The function \texttt{assignDetectionsToTracks} solves the assignment problem using Munkres's variant of the Hungarian algorithm. This function takes the cost matrix and the cost of unassigned states/sensors as input, and returns the indexes of assigned and unassigned states/sensors as output.

As the final comment, recall that we  consider sensor cost optimization only for structurally cyclic systems. For systems which are not structurally cyclic, other than parent SCCs, another type of observationally equivalent set emerges, known as \textit{contraction} \citep{jstsp14}. Number of contractions equals to the number of unmatched nodes in the system digraph which in turn equals to system rank deficiency. Contractions and Parent SCCs determine the number of necessary states for system observability. The key point is that, unlike Parent SCCs which are separate sets, contractions may share state nodes with each other and with parent SCCs\citep{jstsp14}. This implies that the problem cannot be generally reformulated as LSAP and  Problem Formulation 3 is only valid and exact for structurally cyclic systems. In general systems, particularly in structurally rank-deficient systems, a combination of assignment problem and greedy algorithms may need to be applied, which is the direction of future research.

\section*{Acknowledgement}
The first author would like to thank Professor Usman Khan from Tufts University for his helpful suggestions and feedback on this paper.

\bigskip

\end{document}